\documentclass[runningheads]{llncs}

\usepackage{booktabs,multirow, subfigure, tabularx}
\usepackage{bm, array, amsmath, amsfonts, graphicx, amssymb}
\usepackage{sidecap}
\sidecaptionvpos{figure}{t}
\usepackage[ruled,linesnumbered]{algorithm2e}
\usepackage[colorlinks=true, allcolors=blue]{hyperref}
\newcommand{\furl}[1]{\footnote{\url{http://#1}}}

\begin{document}
\title{Discovering Bayesian Market Views \\for Intelligent Asset Allocation\thanks{A camera-ready version to appear at \emph{ECML-PKDD 2018}.}}
\titlerunning{Discovering Bayesian Market Views} % abbreviated title (for running head)
\author{Frank Z. Xing\inst{1} \and Erik Cambria\inst{1} \and Lorenzo Malandri\inst{2} \and Carlo Vercellis\inst{2}}
\authorrunning{F. Xing et al.} 
\institute{School of Computer Science and Engineering, Nanyang Technological University\\
\and
Data Mining and Optimization Research Group, Politecnico di Milano\\
\email{\{zxing001,cambria\}@ntu.edu.sg \{lorenzo.malandri,carlo.vercellis\}@polimi.it}}
\maketitle

\begin{abstract}
Along with the advance of opinion mining techniques, public mood has been found to be a key element for stock market prediction. However, how market participants' behavior is affected by public mood has been rarely discussed. Consequently, there has been little progress in leveraging public mood for the asset allocation problem, which is preferred in a trusted and interpretable way. In order to address the issue of incorporating public mood analyzed from social media, we propose to formalize public mood into market views, because market views can be integrated into the modern portfolio theory. In our framework, the optimal market views will maximize returns in each period with a Bayesian asset allocation model. We train two neural models to generate the market views, and benchmark the model performance on other popular asset allocation strategies. Our experimental results suggest that the formalization of market views significantly increases the profitability ($5\%$ to $10\%$ annually) of the simulated portfolio at a given risk level. 
\end{abstract}

\section{Introduction} \label{intro}
Sales and macroeconomic factors are some of the driving forces behind stock movements but there are many others. For example, the subjective views of market participants also have important effects. Along with the growing popularity of social media in the past decades, people tend to rapidly express and exchange their thoughts and opinions~\cite{10-OCB-T}. As a result, the importance of their views has dramatically risen~\cite{16-CE-T}. Currently, stock movements are considered to be essentially affected by new information and the beliefs of investors~\cite{15-LQ-T}.

Meanwhile, sentiment analysis has emerged as a new tool for analyzing the opinions shared on social media~\cite{17-CE-B}. It is a branch of affective computing research that aims to classify natural language utterances as either positive or negative, but sometimes also neutral~\cite{17-CI-T}. In the financial domain, sentiment analysis is frequently used to obtain a data stream of public mood toward a company, stock, or the economy. 
Public mood is the aggregation of individual sentiments which can be obtained and estimated from various sources, such as stock message boards~\cite{04-AW-F,15-NT-T}, blogs, newspapers, and really simple syndication (RSS) feeds~\cite{10-ZW-T}. 

Recently, Twitter has become a dominant microblogging platform on which many works rely for their investigations, such as~\cite{13-SJ-S,15-RG-S,15-NM-F}. Many previous studies support the claim that public mood helps to predict the stock market. For instance, the fuzzy neural network model considering public mood achieves high directional accuracy in predicting the market index. The mood time series is also proved a Granger cause of the market index~\cite{11-BJ-T}. Si et al.~build a topic-based sentiment time series and predict the market index better with a vector autoregression model to interactively link the two series~\cite{13-SJ-T}. The Hurst exponents also suggest a long-term dependency for time series of mood extracted form financial news, similar to many market indices~\cite{17-CS-T}.

Despite the important role in stock market prediction, we assume that public mood does not directly effect the market: it does \emph{indirectly} through market participants' views. The actions taken by market participants as agents, are dependent on their own views, and their knowledge about other agents' views. The changes of asset prices are the consequences of such actions. These assumptions are very different from econometric research using productivity, equilibrium, and business cycle models~\cite{13-AGM-E}, but closer to agent-based models~\cite{08-HC-E}. However, the mechanism of how market views are formed from public mood is heavily overlooked even in the latter case. An intuitive hypothesis could be: the happier the public mood, the higher the stock price. In the real-world market, however, this relationship is far more complicated. Therefore, existing superficial financial applications of AI do not appear convincing to professionals.

In this paper, we attempt to fill this gap by proposing a method for incorporating public mood to form market views computationally. To validate the quality of our views, we simulate the trading performance with a constructed portfolio. The key \emph{contributions} of this paper can be summarized as follows:

\begin{enumerate}
\item We introduce a stricter and easier-to-compute definition of the market views based on a Bayesian asset allocation model. We prove that our definition is compatible, and has the equivalent expressiveness as the original form. 
%\vspace{\baselineskip}
\item We propose a novel online optimization method to estimate the expected returns by solving temporal maximization problem of portfolio returns.
%\vspace{\baselineskip}
\item Our experiments show that the portfolio performance with market views blending public mood data stream is \emph{better} than directly training a neural trading model without views. This superiority is robust for different models selected with the right parameters to generate market views.
\end{enumerate}
The remainder of the paper is organized as follows: Sect.~\ref{bkg} explains the concept of Bayesian asset allocation; following, we describe the methodologies developed for modeling market views in Sect.~\ref{mtd}; we evaluate such methodologies by running trading simulations with various experimental settings in Sect.~\ref{rsts} and show the interpretability of our model with an example in Sect.~\ref{story}; finally, Sect.~\ref{con} concludes the paper and describes future work.

\section{Bayesian Asset Allocation} \label{bkg}
The portfolio construction framework~\cite{52-MH-F} has been a prevalent model for investment for more than half a century. Given the an amount of initial capital, the investor will need to allocate it to different assets. Based on the idea of trading-off between asset returns and the risk taken by the investor, the mean-variance method proposes the condition of an efficient portfolio as follows~\cite{52-MH-F,01-SM-M}:
\begin{align} \label{mpt}
& {\text{maximize}}
& & \begin{matrix} \text{return item} \\\overbrace{\sum^N_{i=1} \mu_i w_i} \end{matrix}
\begin{matrix} \\- \end{matrix}
\begin{matrix} \text{risk item} \\ \overbrace{\frac{\delta}{2}\sum^N_{i=1}\sum^N_{j=1} w_i \sigma_{ij} w_j} \end{matrix}\\
& \text{subject to}
& & \sum^N_{i=1} w_i =1,\ i=1,2,...,N.\ \ \ \ w_i\geq 0.\nonumber
\end{align}
where $\delta$ is an indicator of risk aversion, $w_i$ denotes the weight of the corresponding asset in the portfolio, $\mu_i$ denotes the expected return of asset $i$, $\sigma_{ij}$ is the covariance between returns of asset $i$ and $j$. The optimized weights of an efficient portfolio is therefore given by the first order condition of Eq.~\ref{mpt}: 
\begin{equation} \label{hfoc}
w^{\ast} = (\delta \Sigma)^{-1} \mu
\end{equation}
where $\Sigma$ is the covariance matrix of asset returns and $\mu$ is a vector of expected returns $\mu_i$. At the risk level of holding $w^{\ast}$, the efficient portfolio achieves the maximum combinational expected return.

However, when applying this mean-variance approach in real-world cases, many problems are faced. For example, the two moments of asset returns are difficult to estimate accurately~\cite{17-SW-T}, as they are non-stationary time series. The situation is worsened by the fact that, the Markowitz model is very sensitive to the estimated returns and volatility as inputs. The optimized weights can be very different because of a small error in $\mu$ or $\Sigma$.
To address the limitation of the Markowitz model, a Bayesian approach that integrates the additional information of investor's judgment and the market fundamentals was proposed by Black and Litterman~\cite{91-BF-F}. In the Black-Litterman model, the expected returns $\mu_{BL}$ of a portfolio is inferred by two antecedents: the equilibrium risk premiums $\Pi$ of the market as calculated by the capital asset pricing model (CAPM), and a set of views on the expected returns of the investor. 

The Black-Litterman model assumes that the equilibrium returns are normally distributed as $r_{eq}\sim \mathcal{N}(\Pi, \tau\Sigma)$, where $\Sigma$ is the covariance matrix of asset returns, $\tau$ is an indicator of the confidence level of the CAPM estimation of $\Pi$. The market views on the expected returns held by an investor agent are also normally distributed as $r_{views}\sim \mathcal{N}(Q, \Omega)$. 

Subsequently, the posterior distribution of the portfolio returns providing the views is also Gaussian. If we denote this distribution by $r_{BL}\sim \mathcal{N}(\bar{\mu}, \bar{\Sigma})$, then $\bar{\mu}$ and $\bar{\Sigma}$ will be a function of the aforementioned variables (see Fig.~\ref{bld}). 
\begin{equation} \label{func}
\left[\bar{\mu}, \bar{\Sigma}\right] = f (\tau, \Sigma, \Omega, \Pi, Q)
\end{equation}
\begin{SCfigure}
\caption{The posterior distribution of the expected returns as in the Black-Litterman model, which has a mean between two prior distributions and a variance less than both of them.}\label{bld}
\includegraphics[width=0.5\textwidth]{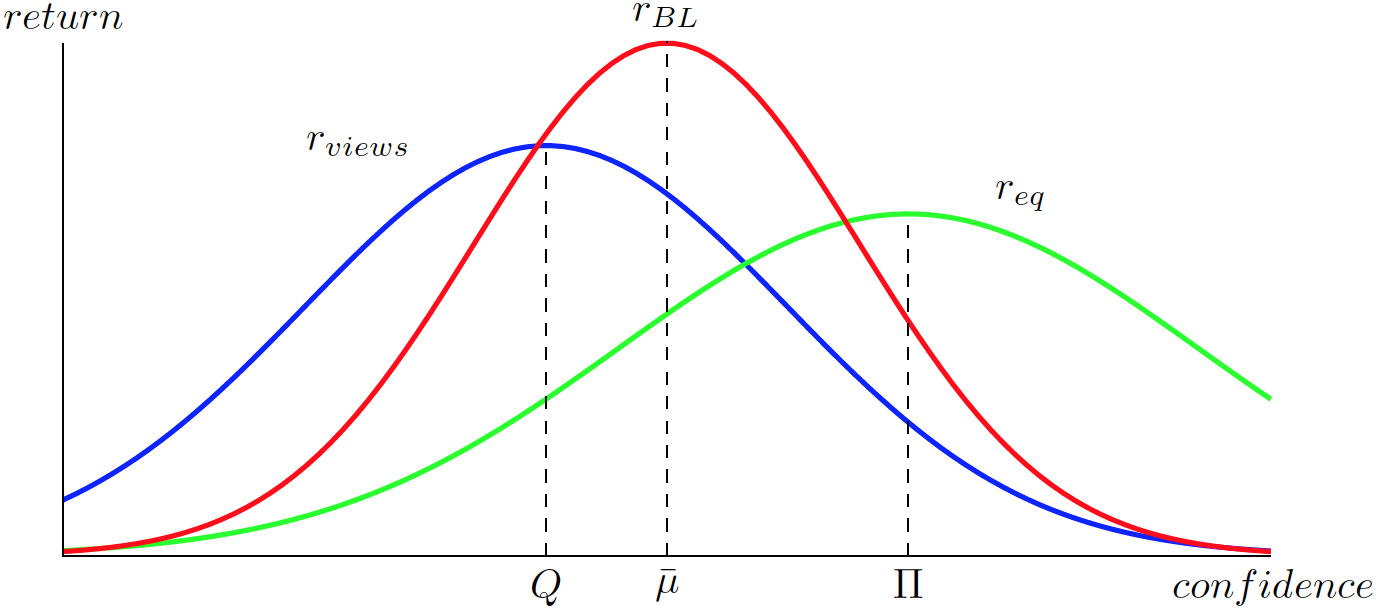}
\end{SCfigure}

The function can be induced from applying Bayes' theorem on the probability density function of the posterior expected returns:
\begin{equation} \label{func}
pdf(\bar{\mu}) = \frac{pdf(\bar{\mu}|\Pi)\ pdf(\Pi)}{pdf(\Pi|\bar{\mu})}
\end{equation}\\
Then, the optimized Bayesian portfolio weights have a similar form to Eq.~\ref{hfoc}, only substituting $\Sigma$ and $\mu$ by $\bar{\Sigma}$ and $\bar{\mu}$:
\begin{equation}
w^{\ast}_{BL} = (\delta \bar{\Sigma})^{-1} \bar{\mu}.
\end{equation}

The most common criticism of the Black-Litterman model is the subjectivity of investor's views. In other words, the model resorts to the good quality of the market views, while it leaves the question of how to actually form these views unanswered. In Sect.~\ref{mtd}, we will investigate the possibility of automatically formalizing the market views from public mood distilled from the Web and the maximization of portfolio returns for each time period.  

\section{Methodologies} \label{mtd}
\subsection{Modeling Market Views}
The Black-Litterman model defines a view as a statement that the expected return of a portfolio has a normal distribution with mean equal to $q$ and a standard deviation given by $\omega$. This hypothetical portfolio is called a~\emph{view portfolio}~\cite{02-HG-F}. In practice, there are two intuitive types of views on the market, termed \emph{relative views} and \emph{absolute views}, that we are especially interested in. Next, we introduce the formalization of these two types of views.

Because the standard deviation $\omega$ can be interpreted as the confidence of expected return of the view portfolio, a \emph{relative view} takes the form of ``I have $\omega_1$ confidence that asset $x$ will outperform asset $y$ by $a\%$ (in terms of expected return)''; an \emph{absolute view} takes the form of ``I have $\omega_2$ confidence that asset $z$ will outperform the (whole) market by $b\%$''. Consequently, for a portfolio consisting of $n$ assets, a set of $k$ views can be represented by three matrices $P_{k,n}$, $Q_{k,1}$, and $\Omega_{k,k}$. 

$P_{k,n}$ indicates the assets mentioned in views. The sum of each row of $P_{k,n}$ should either be $0$ (for relative views) or $1$ (for absolute views); $Q_{k,1}$ is a vector comprises expected returns for each view. Mathematically, the confidence matrix $\Omega_{k,k}$ is a measure of covariance between the views. 
The Black-Litterman model assumes that the views are independent of each other, so the confidence matrix can be written as $\Omega=diag(\omega_1, \omega_2, ..., \omega_n)$. In fact, this assumption will not affect the expressiveness of the views as long as the $k$ views are compatible (not self-contradictory). Because when $\Omega_{k,k}$ is not diagonal, we can always do spectral decomposition: $\Omega=V\Omega^{\Lambda}V^{-1}$. Then we write the new mentioning and new expected return matrices as $P^{\Lambda}=V^{-1}P$, $Q^{\Lambda}=V^{-1}Q$, where $\Omega^{\Lambda}$ is diagonal.
Under these constructions, we introduce two important properties of the view matrices in Theorem~\ref{thm:civ} and Theorem~\ref{thm:uavm}.

\begin{theorem} [Compatibility of Independent Views]\label{thm:civ}
Any set of independent views are compatible. 
\end{theorem}
\begin{proof}
Compatible views refer to views that can hold at the same time. For example, \{asset $x$ will outperform asset $y$ by $3\%$, asset $y$ will outperform asset $z$ by $5\%$, asset $x$ will outperform asset $z$ by $8\%$\} is compatible. However, if we change the third piece of view to ``asset $z$ will outperform asset $x$ by $8\%$", the view set becomes self-contradictory. Because the third piece of view is actually a deduction from the former two, the view set is called ``not independent".

Assume there is a pair of incompatible views $\{p, q\}$ and $\{p, q^{\prime}\}$, $q\neq q^{\prime}$. Both views are either explicitly stated or can be derived from a set of $k$ views. 
Hence, there exist two different linear combinations, such that: 
\begin{subequations}{} 
\begin{equation}
\sum\limits_{i=1}^{k} a_i p_i = p \qquad\sum\limits_{i=1}^{k} a_i q_i = q\nonumber
\end{equation}
\begin{equation}
\sum\limits_{i=1}^{k} b_i p_i = p \qquad\sum\limits_{i=1}^{k} b_i q_i = q^{\prime}\nonumber
\end{equation}
\end{subequations}
where $(a_i-b_i)$ are not all zeros.

Thus, we have $\sum\limits_{i=1}^{k} (a_i-b_i) p_i =\bf{0}$, which means that matrix $P$ is rank deficient and the $k$ views are not independent.
According to the law of contrapositive, the statement ``all independent view sets are compatible'' is true. \qed
\end{proof}

\begin{theorem} [Universality of Absolute View Matrix]\label{thm:uavm}
Any set of independent relative and absolute views can be expressed with a non-singular absolute view matrix. 
\end{theorem}
\begin{proof}
Assume a matrix $P$ with $r$ relative views and $(k-r)$ absolute views.

\begin{centering}
$P_{k,n} = 
 \begin{bmatrix}
 p_{1,1} & p_{1,2} & \cdots & p_{1,n} \\
 \vdots & \vdots & \ddots & \vdots \\
 p_{r,1} & p_{r,2} & \cdots & p_{r,n} \\
 \vdots & \vdots & \ddots & \vdots \\
 p_{k,1} & p_{k,2} & \cdots & p_{k,n} 
 \end{bmatrix}$\\
\end{centering}

The corresponding return vector is $Q=(q_1, q_2, \dots, q_k)$, the capital weight vector for assets is $w=(w_1, w_2, \dots, w_k)$. Hence, we can write $(r+1)$ equations with regard to $r$ new variables $\{q_1^{\prime}, q_2^{\prime}, ..., q_r^{\prime}\}$, where $j=1, 2, ..., r$:
\begin{subequations}{} 
\begin{equation}
1+ q_j^{\prime}=\sum\limits_{i\neq j}^{r}(1+q_i^{\prime}) \frac{w_i}{\sum\limits_{s\neq j} w_s} (1+q_j)\nonumber
\end{equation}
\begin{equation}
\sum\limits_{i=1}^{r}q_i^{\prime}w_i+\sum\limits_{i=r+1}^{k}q_i w_i=Qw^\intercal \nonumber
\end{equation}
\end{subequations}
If we consider $\{asset_{r+1},\dots, asset_k\}$ to be one asset, return of this asset is decided by $P_{r,n}$. Hence, $r$ out of the $(r+1)$ equations above are independent.

According to Cramer's rule, there exists a unique solution $Q^{\prime}=(q_1^{\prime}, q_2^{\prime}, \dots, q_r^{\prime},\\ q_{r+1}, \dots, q_k)$ to the aforementioned $(r+1)$ equations, such that view matrices $\{P^{\prime}, Q^{\prime}\}$ is equivalent to view matrices $\{P, Q\}$ for all the assets considered, where

\begin{centering}
$P^{\prime}_{k,n} = 
 \begin{bmatrix}
 1 & 0 & \cdots & 0 \\
 \vdots & \vdots & \ddots & \vdots \\
 0 & p_{r,r}=1 & \cdots & 0 \\
 \vdots & \vdots & \ddots & \vdots \\
 p_{k,1} & p_{k,2} & \cdots & p_{k,n} 
 \end{bmatrix}$.\\
\end{centering}
Now, $P^{\prime}_{k, n}$ only consists of absolute views. By deleting those dependent views, we can have a non-singular matrix that only consists of absolute views and is compatible. \qed
\end{proof}
Given Theorem~\ref{thm:civ} and Theorem~\ref{thm:uavm}, without loss of generality, we can use the following equivalent yet stricter definition of market views to reduce computational complexity. 

\begin{definition}
Market views on $n$ assets can be represented by three matrices $P_{n,n}$, $Q_{n,1}$, and $\Omega_{n,n}$, where $P_{n,n}$ is an identity matrix; $Q_{n,1} \in \mathbb{R}^n$; $\Omega_{n,n}$ is a nonnegative diagonal matrix. 
\end{definition}

\subsection{The Confidence Matrix}
In the most original form of the Black-Litterman model, the confidence matrix $\Omega$ is set manually according to investors' experience. Whereas in the numerical example given by~\cite{02-HG-F}, the confidence matrix is derived from the equilibrium covariance matrix:
\begin{equation}
\hat{\Omega}_0 = diag(P(\tau\Sigma)P^{\prime})
\end{equation}
This is because $P(\tau\Sigma)P^{\prime}$ can be understood as a covariance matrix of the expected returns in the views as well. Using our definition, it is easier to understand this estimation, because $P$ is an identity matrix, $P(\tau\Sigma)P^{\prime}$ is already diagonal. The underlying assumption is that the variance of an absolute view on asset $i$ is proportional to the volatility of asset $i$. In this case, the estimation of $\Omega$ utilizes past information of asset price volatilities. 

\subsection{Optimal Market Views}
We obtain the optimal market views $\{P, Q, \Omega\}$ in a hybrid way, first we adopt the confidence matrix $\hat{\Omega}_0$, then $Q$ can be derived from the inverse optimization problem using the Black-Litterman model.

We start from the optimal portfolio weights that maximize the portfolio returns for each period $t$. Obviously, without short selling and transaction fees, one should re-invest his whole capital daily to the fastest-growing asset in the next time period. 

The optimal holding weights for each time period $t$ thus take the form of a one-hot vector, where $\oslash$ and $\odot$ denote element-wise division and product:
\begin{equation} 
w_t^{\ast}=\mathrm{argmax}\;\; w_t \oslash price_{t} \odot price_{t+1}
\end{equation}
Let this $w_t^{\ast}$ be the solution to Eq.~\ref{mpt}, we will have:
\begin{equation} \label{optw}
w_t^{\ast} = (\delta \bar{\Sigma}_t)^{-1} \bar{\mu}_t
\end{equation}
where the Black-Litterman model gives\footnote{The proof of Eq.~\ref{blt1} and~\ref{blt2} can be found from the appendix of~\cite{00-SS-F}.}:
\begin{align} 
&\bar{\Sigma}_t=\Sigma_t+[(\tau\Sigma_t)^{-1}+P^{\prime}\hat{\Omega}^{-1}_tP]^{-1}\label{blt1}\\
&\bar{\mu}_t= [(\tau\Sigma_t)^{-1}+P^{\prime}\hat{\Omega}^{-1}_tP]^{-1}[(\tau\Sigma_t)^{-1}\Pi_t+P^{\prime}\hat{\Omega}^{-1}_t Q_t] \label{blt2}
\end{align}

According to Eq.~\ref{optw},~\ref{blt1}, and~\ref{blt2}, the optimal expected returns for our market views for each period $t$ is:
\begin{equation} \label{qtstar}
\begin{split}
Q_t^{\ast} &=\hat{\Omega}_{0,t}\big\{[\,(\tau\Sigma_t)^{-1}+P^{\prime}\hat{\Omega}_{0,t}^{-1}P\,]\,\bar{\mu}_t-(\tau\Sigma_t)^{-1}\Pi_t\big\}\\
&=\delta[\,\hat{\Omega}_{0,t}(\tau\Sigma_t)^{-1}+\mathbb{I}\,]\,\bar{\Sigma}_tw_t^{\ast}-\hat{\Omega}_{0,t}(\tau\Sigma_t)^{-1}\Pi_t\\
&=\delta[\,\hat{\Omega}_{0,t}(\tau\Sigma_t)^{-1}+\mathbb{I}\,]\,[\,\Sigma_t+[(\tau\Sigma_t)^{-1}+\hat{\Omega}^{-1}_t]^{-1}\,]w_t^{\ast}\\
&\quad-\hat{\Omega}_{0,t}(\tau\Sigma_t)^{-1}\Pi_t
\end{split}
\end{equation}

\subsection{Generating Market Views with Neural Models}
Eq.~\ref{qtstar} provides a theoretical perspective on determining the expected return of optimal market views. However, computing $w_t^{\ast}$ requires future asset prices, which is not accessible. Therefore, the feasible approach is to learn approximating $Q_t^{\ast}$ with historical data and other priors as input. We use the time series of asset prices, trading volumes, and public mood data stream to train neural models ($nn$) for this approximation problem of optimal market views:
\begin{equation} \label{appx}
\hat{Q}_t = nn(prices, volumes, sentiments;\, Q_t^{\ast})
\end{equation}
We denote the time series of asset prices $price_{t-k}, price_{t-k+1}, ..., price_t$ by a lag operator $\mathcal{L}^{0\sim k}price_t$. The notation of trading volumes follows a similar form. Then the model input at each time point: $[\mathcal{L}^{0\sim k}price_t, \mathcal{L}^{0\sim k}volume_t, sentiment_t,\\ capital_t]$ can be denoted by $[p,v,s,c]_t$ in short.

Two types of neural models, including a neural-fuzzy approach and a deep learning approach are trained for comparison. Fig.~\ref{dataflow} provides an illustration of the online training process using a long short-term memory (LSTM) network, where $\hat{Q}$ is the output. 
\begin{figure}[htbp]
\centering
\def\layersep{1.5cm}
\includegraphics[width=0.772\textwidth]{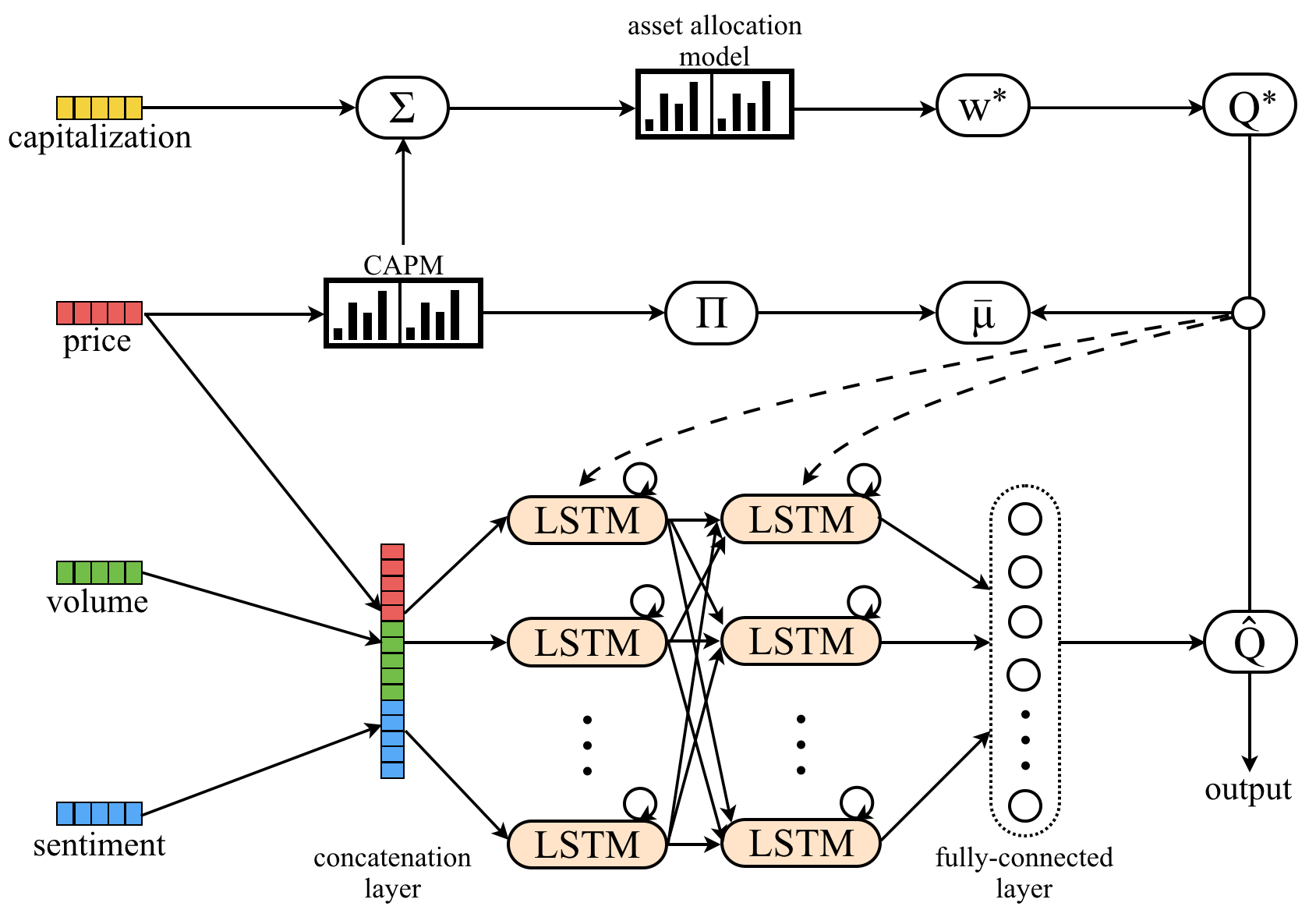}
\caption{Model training process (LSTM) with/without sentiment information.} \label{dataflow}
\end{figure}

\subsubsection{Dynamic evolving neural-fuzzy inference system (DENFIS)} is a neural network model with fuzzy rule nodes~\cite{02-KN-T}. The partitioning of which rule nodes to be activated is dynamically updated with the new distribution of incoming data. This evolving clustering method (ECM) features the model with stability and fast adaptability. Comparing to many other fuzzy neural networks, DENFIS performs better in modeling nonlinear complex systems~\cite{me.17-a}.

Considering the financial market as a real-world complex system, we learn the first-order Takagi-Sugeno-Kang type rules online. Each rule node has the form of:
\begin{align*} 
&\emph{IF $\mathcal{L}^{0\sim k}attribute_{t,i}=pattern_i,\, i=1,2,..., N$} \\
&\quad\emph{THEN $\hat{Q_t}=f_{1,2,..., N}([p,v,s]_t)$}
\end{align*}
where we have 3 attributes and $(2^N-1)$ candidate functions to activate. In our implementation of the DENFIS model, all the membership functions are symmetrical and triangular, which can be defined by two parameters $b\pm d/2$. $b$ is where the membership degree equals to $1$; $d$ is the activation range of the fuzzy rule. In our implementation, $b$ is iteratively updated by linear least-square estimator of existing consequent function coefficients.

\subsubsection{LSTM} is a type of recurrent neural network with gated units. This unit architecture is claimed to be well-suited for learning to predict time series with an unknown size of lags and long-term event dependencies. Early attempts, though not very successful~\cite{01-GF-T}, have been made to apply LSTM to time series prediction. It is now recognized that though LSTM cells can have many variants, their performance across different tasks are similar~\cite{16-GK-T}. 

Therefore, we use a vanilla LSTM unit structure. Our implementation of LSTM cells follows the update rules of the input gate, forget gate, and output gate as in Eq.~\ref{lstmup}: 
\begin{equation} \label{lstmup}
\begin{split}
i_t &= \sigma(W_i\cdot[\,h_{t-1},[p,v,s]_t\,]+b_i)\\
f_t &= \sigma(W_f\cdot[\,h_{t-1},[p,v,s]_t\,]+b_f)\\
o_t &= \sigma(W_o\cdot[\,h_{t-1},[p,v,s]_t\,]+b_o)
\end{split}
\end{equation}
where $\sigma$ denotes the sigmoid function, $h_{t-1}$ is the output of the previous state, $W$ is a state transfer matrix, and $b$ is the bias. 

The state of each LSTM cell $c_t$ is updated by:
\begin{equation} 
\begin{split}
&c_t=f_t \odot c_{t-1} + i_t \odot(W_c\cdot[\,h_{t-1},[p,v,s]_t\,]+b_c)\\
&h_{t-1}=o_t\odot \tanh(c_{t-1})
\end{split}
\end{equation}

We make the training process online as well, in a sense that each time a new input is received, we use the previous states and parameters of LSTM cells $[c_{t-1}, \mathbf{W}, \mathbf{b}]$ to initialize the LSTM cells for period $t$.

\section{Experiments} \label{rsts}
To evaluate the quality and effectiveness of our formalization of market views, we run trading simulations with various experimental settings. 

\subsection{Data}
The data used in this study are publicly available on the Web\furl{github.com/fxing79/ibaa}. We obtain the historical closing price of stocks and daily trading volumes from the Quandl API\furl{www.quandl.com/tools/api}; the market capitalization data from Yahoo! Finance; the daily count and intensity of company-level sentiment time series from PsychSignal\furl{psychsignal.com}. The sentiment intensity scores are computed from multiple social media platforms using NLP techniques. Fig.~\ref{150days_aapl} depicts a segment example of the public mood data stream. The market is closed on weekends, so a corresponding weekly cycle of message volume can be observed.

We investigate a window of around 8 years (2800 days). All the time series are trimmed from 2009-10-05 to 2017-06-04. For missing values such as the closing prices on weekends and public holidays, we fill them with the nearest historical data to train the neural models. 
The lagged values we use for both price and trading volume consist of 4 previous days and a moving average of the past 30 days, that is, the input of our neural models takes the form of Eq.~\ref{inputf1} and~\ref{inputf2}:
\begin{figure}[htp]
\centering
\includegraphics[width=\textwidth]{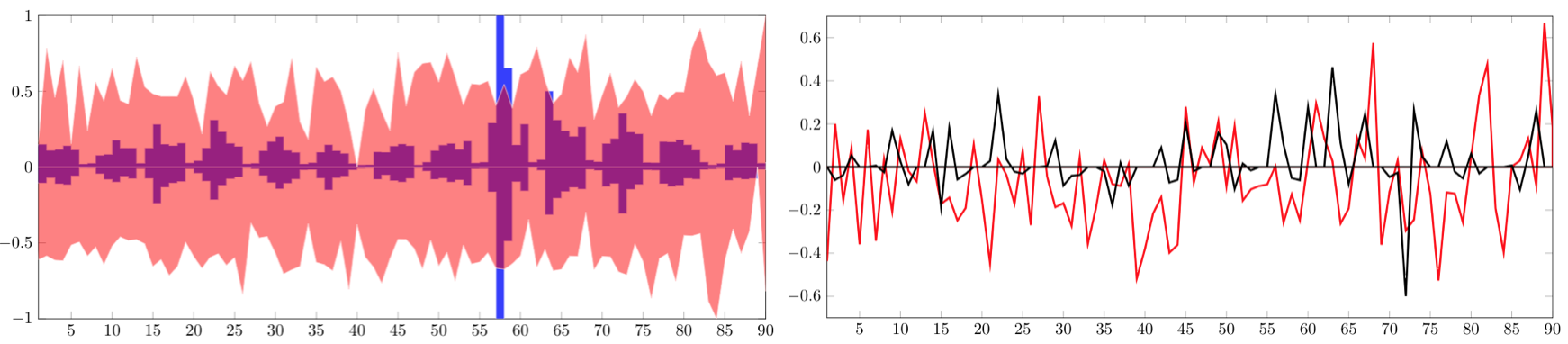}
\caption{The volume of daily tweets filtered by cashtag \texttt{AAPL} (blue, left); average sentiment intensity (red, left); net sentiment polarity (red, right); daily returns (black, right) in a time period of 90 days (2017-03-04 to 2017-06-04). All the series are normalized.} \label{150days_aapl}
\end{figure}
\begin{align}
\mathcal{L}^{0\sim k}price_t&=(p_t, p_{t-1}, p_{t-2}, p_{t-3}, \frac{\sum^{30}_{i=1} p_i}{30})\label{inputf1}\\
\mathcal{L}^{0\sim k}volume_t&=(v_t, v_{t-1}, v_{t-2}, v_{t-3}, \frac{\sum^{30}_{i=1} v_i}{30})\label{inputf2}
\end{align}

\subsection{Trading Simulation}
We construct a virtual portfolio consisting of 5 big-cap stocks: Apple Inc (\texttt{AAPL}), Goldman Sachs Group Inc (\texttt{GS}), Pfizer Inc (\texttt{PFE}), Newmont Mining Corp (\texttt{NEM}), and Starbucks Corp (\texttt{SBUX}). This random selection covers both the NYSE and NASDAQ markets and diversified industries, such as technology, financial services, health care, consumer discretionary etc. During the period investigated, there were two splits: a 7-for-1 split for \texttt{AAPL} on June 9th 2014, and a 2-for-1 split for \texttt{SBUX} on April 9th 2015. The prices per share are adjusted according to the current share size for computing all related variables, however, dividends are not taken into account. 
We benchmark our results with two portfolio construction strategies:

\begin{figure*}[ht]
\centering
\renewcommand{\arraystretch}{1}
\subfigure[No views]{
\label{Fig.sub.1}
\includegraphics[width=0.315\textwidth]{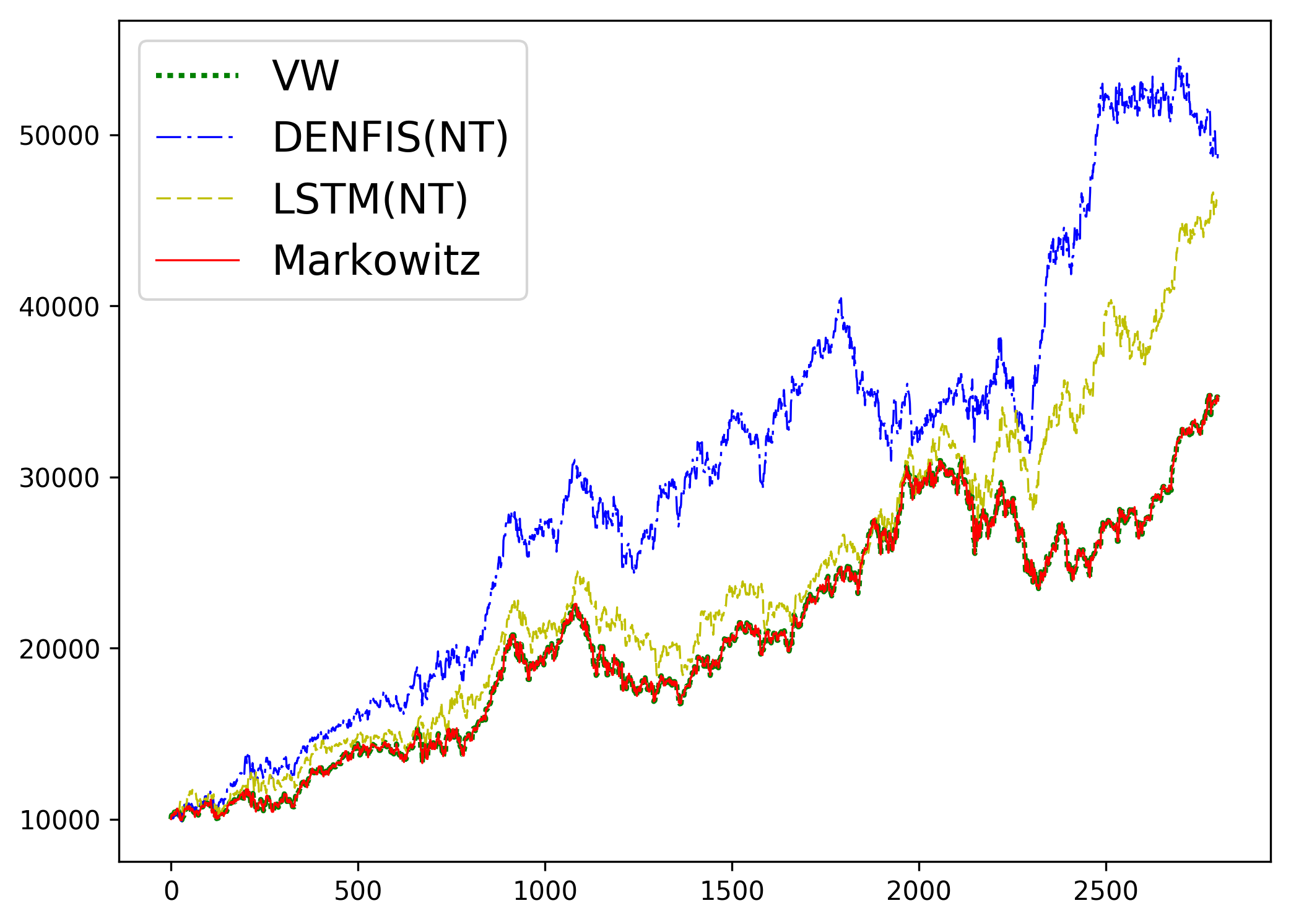}}
\subfigure[Random views]{
\label{Fig.sub.2}
\includegraphics[width=0.315\textwidth]{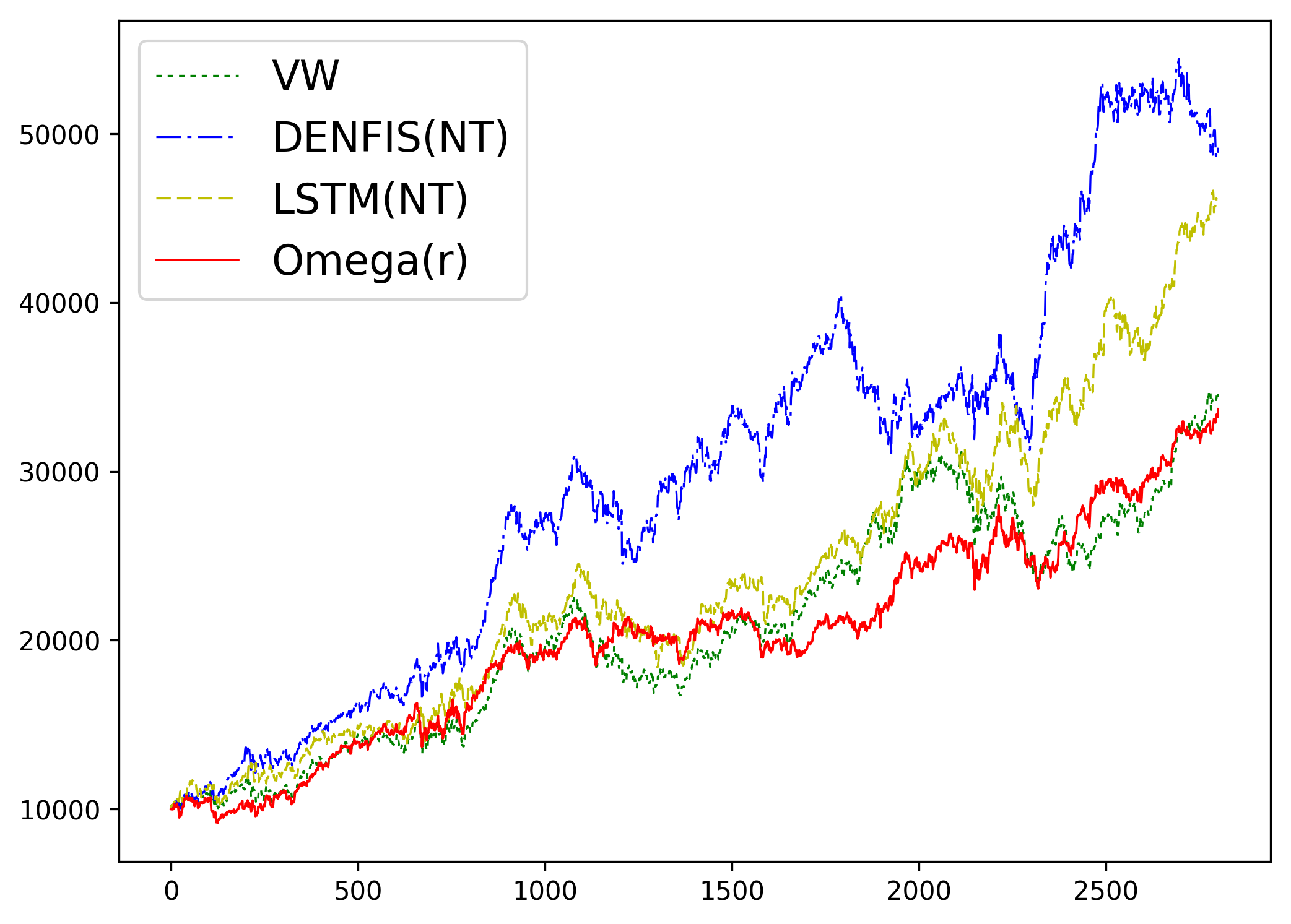}}
\subfigure[BL+sentiment, t=90]{
\label{Fig.sub.3}
\includegraphics[width=0.315\textwidth]{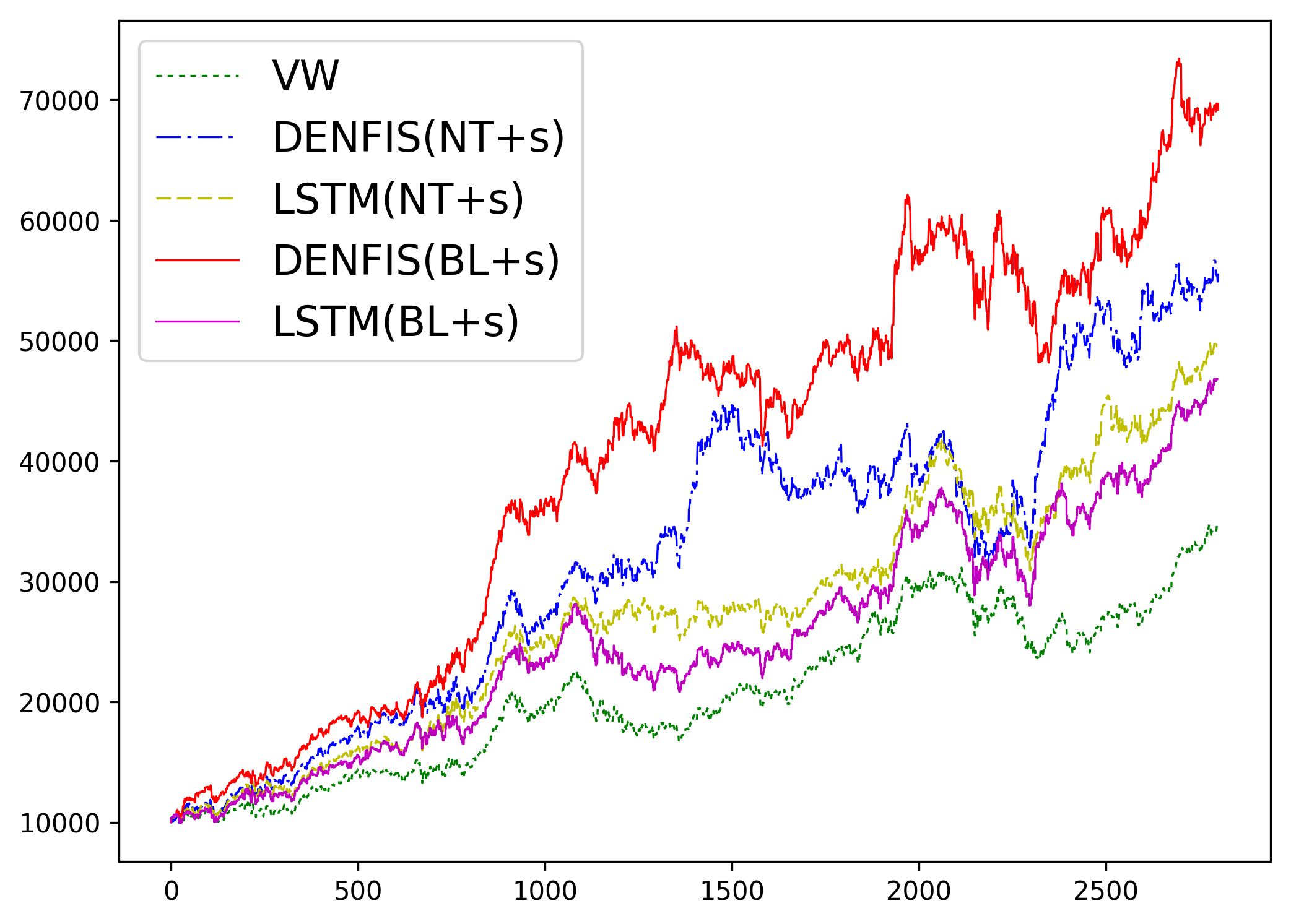}}
\subfigure[DENFIS+sentiment]{
\label{Fig.sub.4}
\includegraphics[width=0.315\textwidth]{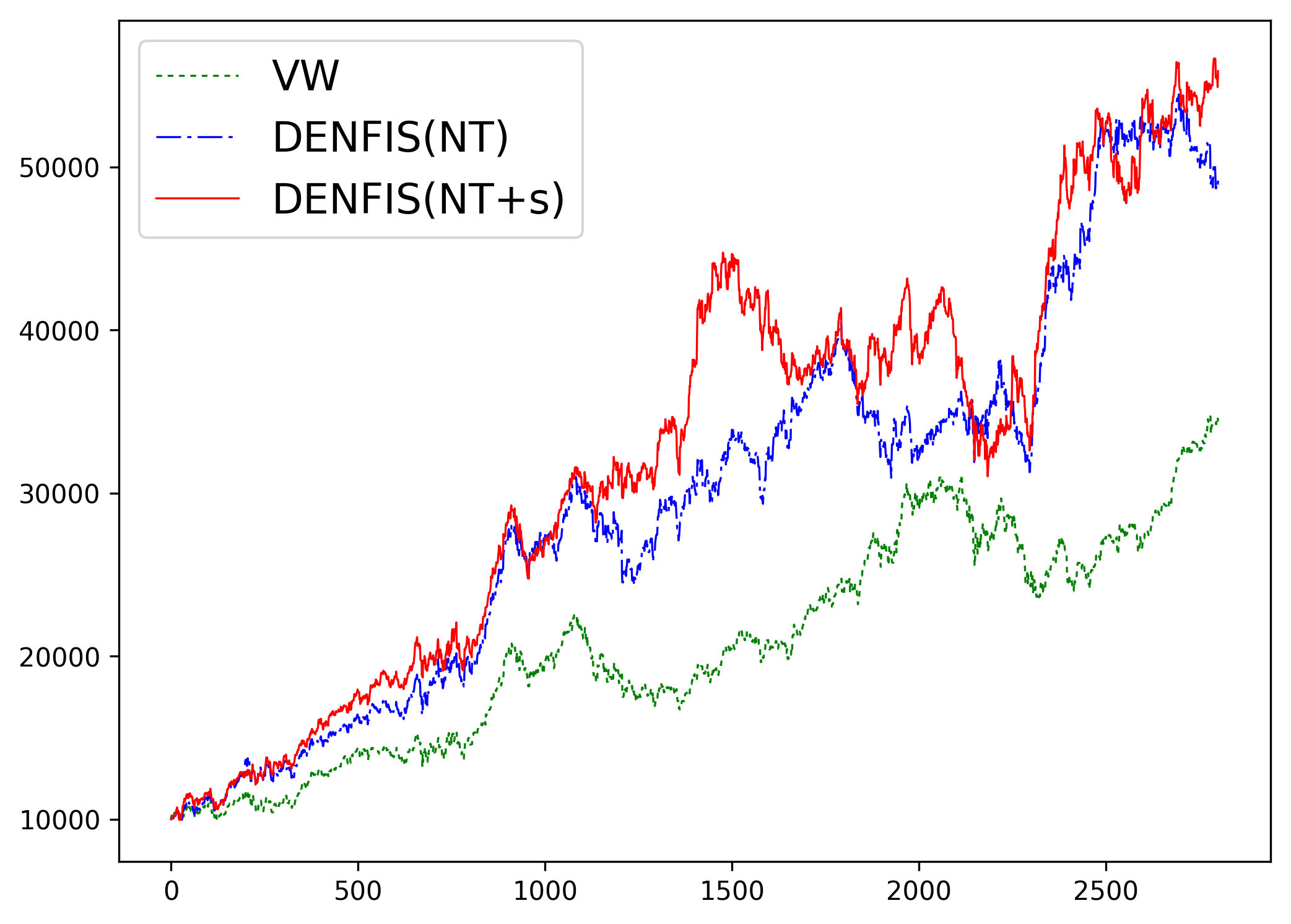}}
\subfigure[LSTM+sentiment]{
\label{Fig.sub.5}
\includegraphics[width=0.315\textwidth]{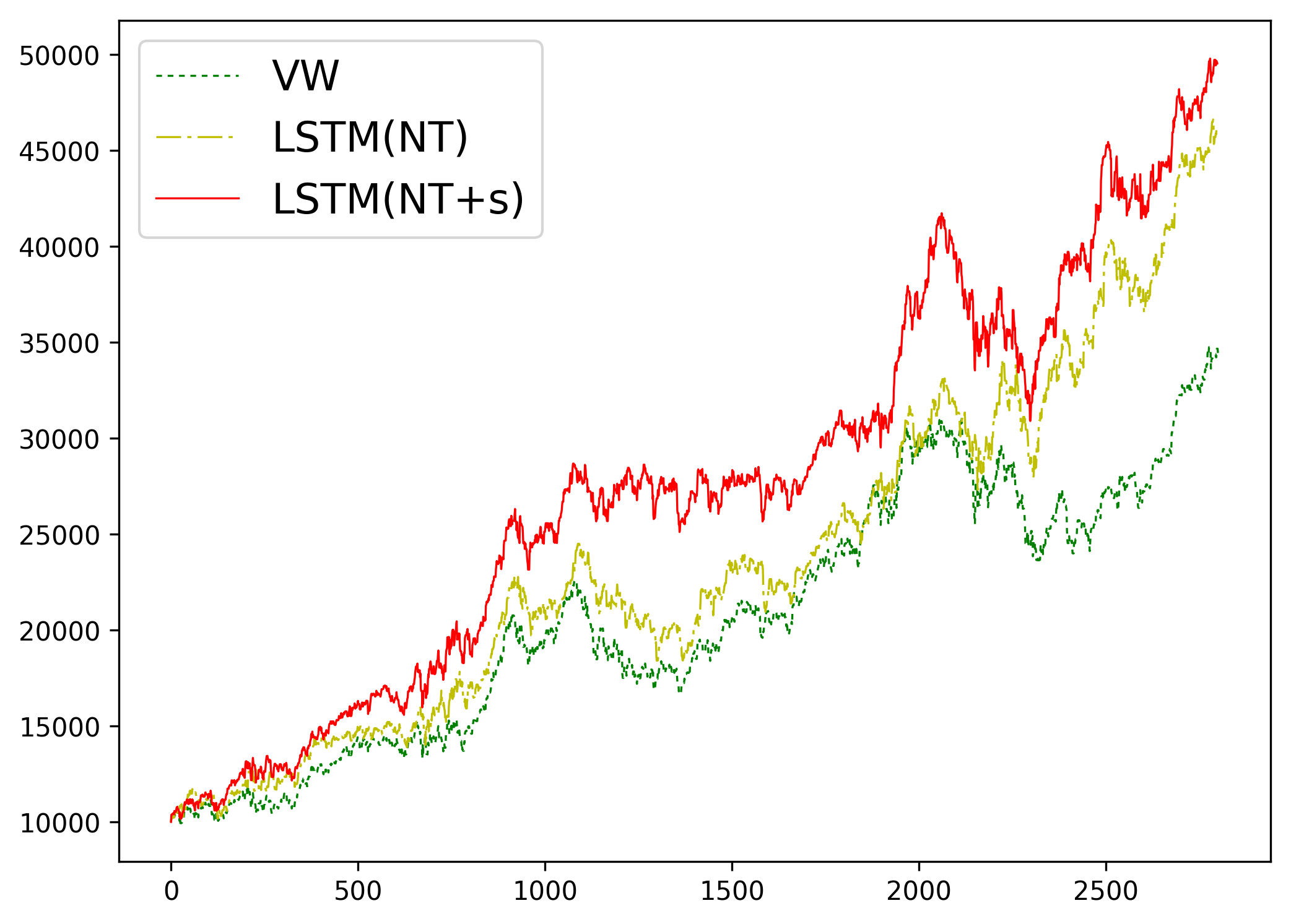}}
\subfigure[BL+sentiment, t=180]{
\label{Fig.sub.6}
\includegraphics[width=0.315\textwidth]{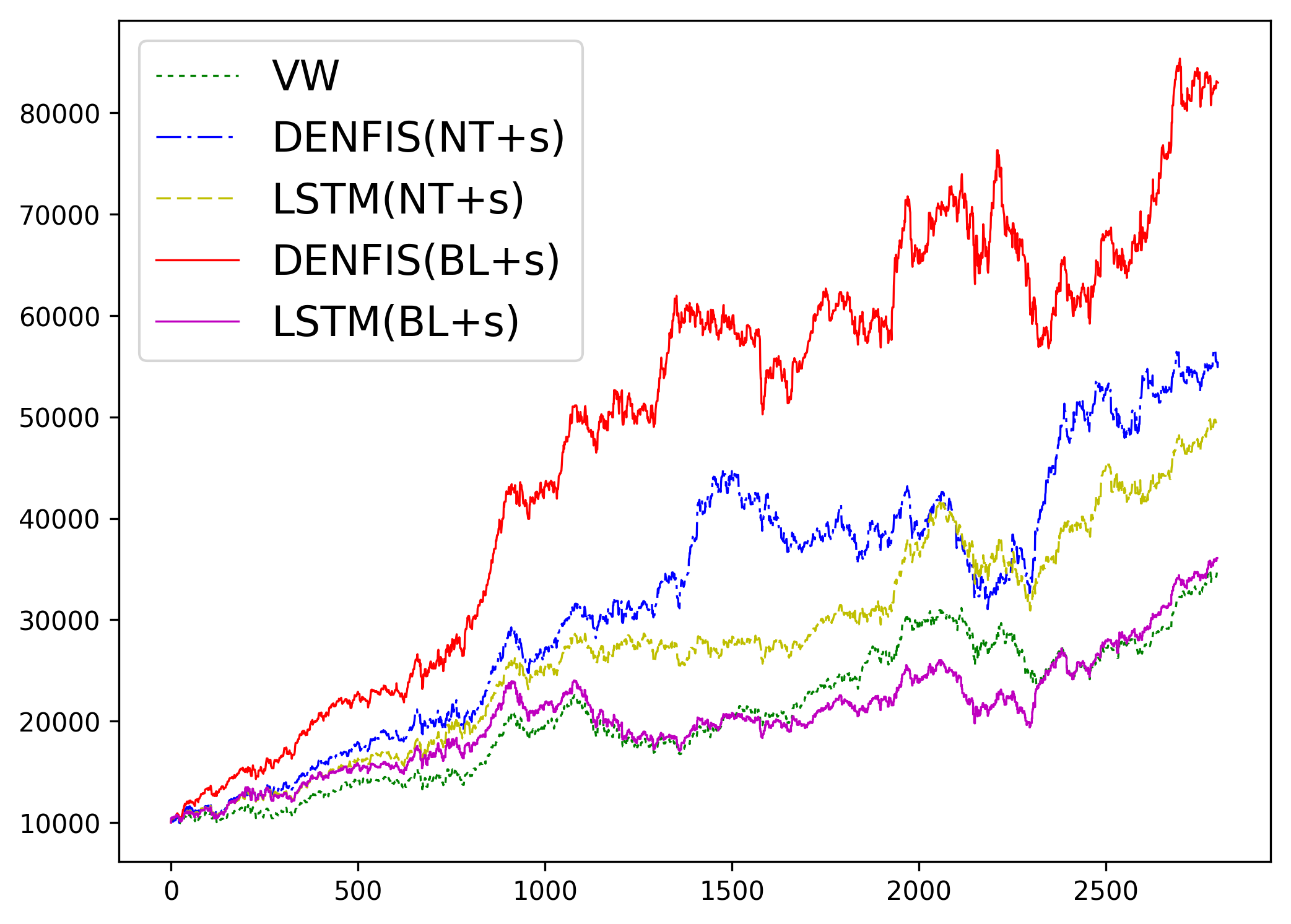}}
\caption{Trading simulation performance with different experimental settings: (x-axis: number of trading days; y-axis: cumulative returns). In particular, we use a timespan of 90 and 180 days for our approach. The performance of neural trading is independent from timespan, accordingly the two neural models are compared in~\ref{Fig.sub.4} and~\ref{Fig.sub.5} respectively for better presentation.}
\label{tpd}
\end{figure*}

\textbf{1) The value-weighted portfolio (VW):} we re-invest daily according to the percentage share of each stock's market capitalization. In this case, the portfolio performance will be the weighted average of each stock's performance. This strategy is fundamental, yet empirical study~\cite{10-FE-F} shows that beating the market even before netting out fees is difficult.

\textbf{2) The neural trading portfolio (NT):} we remove the construction of market views and directly train the optimal weights of daily position with the same input. For this black-box strategy, we can not get any insight on how this output portfolio weight comes about.

In the simulations, we assume no short selling, taxes, or transaction fees, and we assume the portfolio investments are infinitely divisible, starting from $10,000$ dollars. We construct portfolios with no views ($\Omega_{\varnothing}$, in this case the degenerate portfolio is equivalent to Markowitz's mean-variance portfolio using historical return series to estimate covariance matrix as a measure of risk), random views ($\Omega_{r}$), the standard views using the construction of Black-Litterman model ($\Omega_{0}$), with and without our sentiment-induced expected returns ($s$). The trading performances are demonstrated in Fig.~\ref{tpd}.

Following the previous research~\cite{02-HG-F}, we set the risk aversion coefficient $\delta=0.25$ and confidence level of CAPM, $\tau=0.05$. Let the activation range of fuzzy membership function $d=0.21$, we obtain 21 fuzzy rule nodes from the whole online training process of DENFIS. This parameter minimizes the global portfolio weight error. For the second neural model using deep learning, we stack two layers of LSTMs followed by a densely connected layer. Each LSTM layer has 3 units; the densely connected layer has 50 neurons, which is set times larger than the number of LSTM units. We use the mean squared error of vector $Q$ as the loss function and the rmsprop optimizer~\cite{12-TT-U} to train this architecture. We observe fast training error convergence in our experiments.

\subsection{Performance Metrics}
Diversified metrics have been proposed to evaluate the performance of a given portfolio~\cite{04-BM-F,06-HR-F,17-XF-T}. We report four metrics in our experiments.

Root mean square error (RMSE) is a universal metric for approximation problems. It is widely used for engineering and data with normal distribution and few outliers. 
We calculate the RMSE of our realized portfolio weights to the optimal weights:
\begin{equation}
\text{RMSE}= \sqrt{\frac{1}{n} \sum_{i=1}^n \|w_i-\hat{w_i}\|^2}\\
\end{equation}

Annualized return (AR) measures the profitability of a given portfolio. We calculate the geometric mean growth rate per year, which is also referred to as compound annual growth rate (CAGR) for these 2800 days. 

Sharpe ratio (SR) is a risk-adjusted return measure. We choose the value-weighted portfolio as a base, consequently the Sharpe ratio of VW will be $1$:
\begin{equation}
\text{SR}= \frac{\mathbb{E}(R_{portfolio}/R_{VW})}{\sigma (R_{portfolio})/\sigma (R_{VW})}\\
\end{equation}
SR uses the standard deviation of daily returns as the measure of risk. Note that to distinguish between good and bad risk, we can also use the standard deviation of downside returns only~\cite{94-SF-F}. Our results suggest that the Sortino ratios, which are not reported due to page limit, are very close to SRs and lead to the same conclusion. 

The maximum drawdown (MDD) measures the maximum possible percentage loss of an investor:
\begin{equation}
\text{MDD}= \max\limits_{0<t<\tau}\Big\{\frac{Value_t-Value_\tau}{Value_t}\Big\}\\
\end{equation}
Asset allocation strategies with large MDD are exposed to the risk of withdrawal. Table~\ref{metrics} presents the metrics.

\begin{table}[]
\centering
\caption{Performance metrics for various portfolio construction strategies, timespan=90 and 180 days. Top three metrics are in bold.}\label{metrics}
\begin{tabular}{lp{1.7cm}<{\centering}p{1.7cm}<{\centering}p{1.7cm}<{\centering}p{1.7cm}<{\centering}}
\toprule
        & RMSE & SR & MDD(\%) & AR(\%) \\ \midrule
VW       & 0.8908 & 1.00 & 25.81 & 17.49\\
Markowitz90($\Omega_{\varnothing}$)& 0.9062 & 1.00 & 25.81 & 17.51\\
Markowitz180($\Omega_{\varnothing}$)& 0.8957 & 1.00 & 25.82 & 17.45\\
BL90($\Omega_r$) & 0.9932 & 0.90 & \bf{23.47} & 17.17\\
BL180($\Omega_r$)& 0.9717 & 1.06 & \bf{20.59} & 22.31\\
DENFIS(NT)   & 0.9140 & \bf{2.94} & 29.84 & 23.09\\
DENFIS(NT+$s$)  & 0.9237 & \bf{4.35} & \bf{23.07} &\bf{25.16}\\
DENFIS(BL90+$s$) & 0.9424 & 1.52 & 24.44 & \bf{28.69}\\
DENFIS(BL180+$s$)& 0.9490 & \bf{1.58} & 24.19 & \bf{29.49}\\
LSTM(NT)    & \bf{0.8726} & 1.38 & 25.68 & 22.10\\
LSTM(NT+$s$)   & 0.8818 & 1.42 & 25.96 & 23.21\\
LSTM(BL90+$s$)  & \bf{0.8710} & 1.34 & 25.90 & 22.33\\
LSTM(BL180+$s$)& \bf{0.8719} & 1.07 & 24.88 & 17.68 \\ \bottomrule
\end{tabular}
\end{table}

\subsection{Findings}
We have some interesting observations from Fig.~\ref{tpd} and Table~\ref{metrics}. SR and AR are usually considered as the most important, and besides, RMSE and MDD are all very close in our experiments. The correlation between RMSE and the other three metrics is weak, though it is intuitive that if the realized weights are close to the optimal weights, the portfolio performance should be better. On the contrary, the LSTM models seem to overfit as they are trained on the mean squared error of weights or expected return of views~\cite{90-PP-F}. However, as mentioned in Sect.~\ref{intro}, the relationship between weights and daily returns is nonlinear. Therefore, \emph{holding portfolio weights that are close to the optimal weights does not necessarily means that the AR must be higher}. In fact, it is dangerous to use any seemingly reasonable metrics outside the study of asset allocation, such as directional accuracy of price change prediction~\cite{11-BJ-T,16-YA-T}, to evaluate the expected portfolio performance.

The Markowitz portfolio ($\Omega_{\varnothing}$) displays a very similar behavior to the market-following strategy. This is consistent with the inefficacy of the mean-variance approach in practice mentioned by previous studies: holding the Markowitz portfolio is holding the market portfolio. In fact, if the CAPM holds, the market portfolio already reflects the adjustments to risk premiums, that is, fewer market participants will invest on highly risky assets, for this reason their market capitalization will be smaller as well. 

However, the Black-Litterman model does not always guarantee better performance over the Markowitz portfolio. ``Garbage in, garbage out" still holds for this circumstance. Given random views ($\Omega_r$), it can be worse than market-following in terms of both SR and AR. The lesson learned is that \emph{if the investor knows nothing, it is better to hold no views and follow the market than pretending to know something}. 

In our experiments, DENFIS generally performs better than LSTM models, achieving higher SRs and ARs. The reason may be LSTM models adapt faster to the incoming data, whereas financial time series are usually very noisy. The ECM mechanism provides DENFIS models with converging learning rates, which may be beneficial to the stability of memorized rules. However, it is important to note that \emph{the ARs for both neural models improve with the blending of sentiments}. The timespan used to estimate correlation and volatility of assets seems not that critical. DENFIS models perform better with longer timespan, while LSTM models perform better with shorter timespan. The Markowitz portfolio is less affected by timespan. 

\section{A Story} \label{story}
One of the main advantages of our formalization and computing of market views is that some \emph{transparency} is brought to the daily asset reallocation decisions. In most cases, a stock price prediction system based on machine learning algorithms cannot justify ``why he thinks that price will reach that predicted point". Unlike these systems, our method can tell a story of the portfolio to professional investors and advice seekers. Take June 1st 2017 as an example:\\

``{\large O}n June 1st 2017, we observe $164$ positive opinions of polarity $+1.90$, $58$ negative opinions of polarity $-1.77$ on \texttt{AAPL} stock; $54$ positive opinions of polarity $+1.77$, $37$ negative opinions of polarity $-1.53$ on \texttt{GS} stock; $5$ positive opinions of polarity $+2.46$, $1$ negative opinion of polarity $-1.33$ on \texttt{PFE} stock; no opinion on \texttt{NEM} stock; and $9$ positive opinions of polarity $+1.76$, $5$ negative opinions of polarity $-2.00$ on \texttt{SBUX} stock.
Given the historical prices and trading volumes of the stocks, we have $6.29\%$ confidence that \texttt{AAPL} will outperform the market by $-70.11\%$; $23.50\%$ confidence that \texttt{GS} will outperform the market by $263.28\%$; $0.11\%$ confidence that \texttt{PFE} will outperform the market by $-0.50\%$; $1.21\%$ confidence that \texttt{SBUX} will outperform the market by $4.57\%$. Since our current portfolio invests $21.56\%$ on \texttt{AAPL}, $25.97\%$ on \texttt{GS}, $29.43\%$ on \texttt{PFE}, and $23.04\%$ on \texttt{SBUX}, by June 2nd 2017, we should withdraw all the investment on \texttt{AAPL}, $2.76\%$ of the investment on \texttt{GS}, $81.58\%$ of the investment on \texttt{PFE}, and $30.77\%$ of the investment on \texttt{SBUX}, and re-invest them onto \texttt{NEM}."

\section{Conclusion and Future Work} \label{con}
In previous studies which have considered sentiment information for financial forecasting, the role of the investor as a market participant is often absent. In this paper, we present a novel approach to incorporate market sentiment by fusing public mood data stream into the Bayesian asset allocation framework. 

This work is pioneering in formalizing sentiment-induced market views. Our experiments show that the market views provide a powerful method to asset management. We also confirm the efficacy of public mood data stream based on social media for developing asset allocation strategies.

A limitation of this work is that we fixed a portfolio with five assets, though in practice the portfolio selection problem is of equal importance. How to assess the quality of sentiment data is not discussed in this paper as well. We are not at the stage to distinguish or detect opinion manipulation though concern like the open networks are rife with bots does exist. Another limitation is that survivor bias is not taken into account: the risk that assets selected in the portfolio may quit the market or suffer from a lack of liquidity. This problem can be alleviated by only including high quality assets. In the future, we will study examining the quality of sentiment data obtained using different content analysis approaches. We also plan to develop a Bayesian asset allocation model that can deal with market frictions.

\bibliographystyle{splncs04}
\bibliography{market_views}

\end{document}